\documentclass[12pt]{extarticle}
\usepackage{extsizes}
\usepackage{amsmath}
\usepackage{amssymb}
\usepackage{amsthm}
\usepackage{amscd}
\usepackage{amsfonts}
\usepackage{dsfont}
\usepackage{fancyhdr}
\usepackage{enumerate}
\usepackage{youngtab}
\usepackage[utf8]{inputenc}
\usepackage{comment}
\usepackage{cite}
\usepackage{tikz}
\usetikzlibrary{arrows.meta}

\usepackage{subfigure}
\usepackage[margin=1.5in]{geometry}
\usepackage{graphicx}
\usepackage{algorithm}
\usepackage{algpseudocode}
\usepackage{color}
\usepackage[hidelinks]{hyperref}
\usepackage{notation}

\theoremstyle{plain}
\newtheorem{theorem}{Theorem}
\newtheorem*{theoremnocounter}{Theorem}
\newtheorem{corollary}[theorem]{Corollary}

\newtheorem{definition}[theorem]{Definition}
\newtheorem{lemma}[theorem]{Lemma}

\newtheorem{proposition}[theorem]{Proposition}

\begin{document}

\title{Compact graphs and quantum automorphisms}
\author{Pedro Baptista\thanks{Corresponding author: pedro.baptista@dcc.ufmg.br \\ Pedro Baptista and Gabriel Coutinho at the Department of Computer Science, Universidade Federal de Minas Gerais (Brazil). Chris Godsil at the Department of Combinatorics and Optimization, University of Waterloo (Canada). Simon Schmidt at the Department of Computer Science, Ruhr-Universität Bochum (Germany).} \and Gabriel Coutinho \and Chris Godsil \and Simon Schmidt}
\date{\today}
\maketitle
\vspace{-0.8cm}

\begin{abstract} 
Compact graphs are graphs for which the fractional automorphism polytope has no genuinely fractional vertices.  This paper proposes a quantum analogue of this idea by evaluating the fundamental magic unitary of the quantum automorphism group on states, which we show to produce a closed convex set of doubly stochastic matrices sitting between the classical automorphism polytope and the full fractional automorphism polytope. Our main result is that the natural quantum analogue of compactness is classical, that is, a quantum compact graph is classically compact.

We also relate this set to the quantum orbital algebra and obtain a hierarchy of classical and quantum compactness pseudo notions.  The framework recovers familiar consequences of compactness through commutants and suggests quantum analogues of generous transitivity and distance-transitivity.  We also isolate examples and open problems indicating where quantum symmetries may strictly refine the classical compactness theory.
\end{abstract}

\begin{center}
\end{center}

\section{Introduction}

Polyhedral relaxations to combinatorial problems are a common theme in combinatorial optimization, and often the central question is whether the relaxation is tight. Tinhofer \cite{tinhofer1986graph} introduced the notion of Birkhoff graphs, later named compact graphs by Brualdi in \cite{brualdi1988some}, as the class of graphs for which the polytope of doubly stochastic matrices that commute with the adjacency matrix of the graph has only permutation matrices as vertices, that is, graphs for which the fractional automorphism polytope is integral. Tinhofer also showed in \cite{TINHOFER1991253} that deciding isomorphism from a compact graph to another graph can be done in polynomial time. Godsil \cite{GODSIL1997Compact} studied compact graphs and their connection to equitable partitions, distance-regularity and transitivity. An extension of the notion of compactness to permutation groups and coherent (cellular) algebras was given in \cite{EVDOKIMOV1999247}, a topic we will soon connect to our work.

Quantum automorphisms of graphs were introduced by Banica \cite{banica2005quantum}, following Wang's work on quantum permutation groups \cite{Wang1998QuantumSymmetryFiniteSpaces}. Lupini et~al. \cite{LUPINI2020108592} studied their connections to nonlocal games and coherent algebras, and among other things defined what they called the quantum orbital algebra of a graph, an object that will be central to our work.

Classically, compactness says that the fractional automorphism relaxation has no genuinely fractional points: every doubly stochastic matrix commuting with the adjacency matrix is forced to come from a convex combination of automorphisms. Quantum automorphism groups provide a different relaxation of symmetry, one in which the entries of the fundamental magic unitary need not commute. The guiding question of this paper is whether evaluating these noncommutative entries on states produces a useful intermediate object between the classical automorphism polytope and the full fractional automorphism polytope, and whether the finite-dimensional algebra of quantum orbitals captures the corresponding notion of quantum compactness.

\subsection{Quantum automorphism group}

Quantum permutation groups were introduced by Wang \cite{Wang1998QuantumSymmetryFiniteSpaces} as noncommutative analogues of permutation groups. Banica \cite{banica2005quantum} then defined quantum automorphism groups of finite graphs by imposing the condition that the fundamental magic unitary commutes with the adjacency matrix. More recently, quantum automorphisms have been connected to nonlocal games, quantum isomorphism, and coherent configurations; in particular, Lupini, Man\v{c}inska and Roberson \cite{LUPINI2020108592} showed that quantum orbitals form a coherent algebra.

\begin{definition} 
\label{defQuantumAutomorphismGroup}
    Let $X$ be a graph with $n$ vertices. The quantum 
    automorphism group $\qGroup (X)$ is the compact matrix quantum 
    group $(C(\qGroup (X)), u)$, where $C(\qGroup (X))$ is the 
    universal $C^*$-algebra with unit $1$ and generators $u_{ij}$, where $u$ is a magic unitary. The generators are subject to the following relations:

    \begin{enumerate}[(i)]
        \item $u_{ij}^2 = u_{ij}^* = u_{ij}$, for $1 \leq i,j \leq n$,
        \item $\sum_{i=1}^n u_{ij} = 1 = \sum_{i=1}^n u_{ji}$, for $1 \leq j \leq n$,
        \item $u A = A u$, where $A$ is the adjacency matrix of $X$.
    \end{enumerate}

\end{definition}

Note that we call a matrix $u$ \emph{magic unitary} if $(i)$ and $(ii)$ are fulfilled. Furthermore, by $(iii)$ we mean the relations
\begin{align*}
\sum_k u_{ik}A_{kj}=\sum_k A_{ik}u_{kj} \quad \text{for all }1\leq i,j\leq n. 
\end{align*}
 Characters of $C(\qGroup(X))$ recover the ordinary automorphisms of $X$. More generally, states on this $C^*$-algebra give numerical doubly stochastic matrices by evaluating the entries of $u$; these matrices will be the basic convex objects considered below.

\subsection{Matrix algebras containing the adjacency matrix} \label{secCommutants}

Let $A = A(X)$ be the adjacency matrix of a graph $X$, and let $J$ denote the all-ones matrix. The smallest algebra which is interesting to us is the real algebra generated by $A$ and $J$, which we shall denote by
\[ \polAlg = \langle \{A, J\} \rangle.
\]
Note that if $X$ is a connected regular graph, then the dimension of this algebra is equal to the number of distinct eigenvalues of the graph.

The next algebra we consider is the smallest coherent algebra containing $A$. An algebra is called coherent if it is self-adjoint, contains $I$ and $J$, and is closed under Schur (entrywise) multiplication. We shall denote the smallest coherent algebra containing $A$ by $\smallAlg$. Note that it always exists and is unique: the full matrix algebra is coherent, and the intersection of any two coherent algebras containing $A$ is a coherent algebra containing $A$. Furthermore, it holds $\polAlg \subseteq \smallAlg$, and the inclusion can be strict, despite the fact that for most graphs $\polAlg$ is already the full matrix algebra \cite{godsil2012controllable,o2016conjecture}.

We point out that both algebras mentioned above can be computed in polynomial time: this is immediate for $\polAlg$, and $\smallAlg$ can be computed via the $2$-dimensional Weisfeiler--Leman algorithm \cite{weisfeiler1968reduction}.

A third algebra we consider is the coherent algebra generated by the orbitals of the automorphism group of $X$. The orbitals of a group action are the orbits of the induced action on pairs of vertices. The adjacency matrices of the orbitals generate a coherent algebra, which coincides with the algebra of matrices that commute with all permutation matrices corresponding to the automorphisms of $X$. We shall denote this algebra by $\orbAlg$. Note that $\smallAlg \subseteq \orbAlg$, and the inclusion can be strict; see, for instance, the examples and discussion in \cite{GODSIL1997Compact,EVDOKIMOV1999247}.

Following the work of Lupini et~al. \cite[Section 3]{LUPINI2020108592}, we introduce our last coherent algebra: the algebra generated by the quantum orbitals of $X$. The quantum orbitals of $X$ are the orbitals of the action of the quantum automorphism group of $X$ on pairs of vertices. The adjacency matrices of the quantum orbitals generate a (real) coherent algebra, which coincides with the algebra of real matrices \(M\) satisfying \(Mu=uM\), where \(u\) is the fundamental magic unitary. We shall denote this algebra by $\qOrbAlg$.

As a summary, we have the following sequence of algebras containing $A$:
\[
\polAlg \subseteq \smallAlg \subseteq \qOrbAlg \subseteq \orbAlg.
\]
The inclusions may be strict; examples can be extracted from the classical theory of coherent algebras and from the quantum-orbital examples in \cite{LUPINI2020108592,schmidt2020disttrans, schmidt2024quantum}.

Taking the commutant of each of these algebras reverses the inclusions, and we have the following sequence of algebras that commute with $A$:
\begin{equation}
    \orbAlg' \subseteq \qOrbAlg' \subseteq \smallAlg' \subseteq \polAlg'. \label{commutants}
\end{equation}
The next subsection is concerned with certain subsets of these algebras.

As a final remark, all these algebras are finite-dimensional semisimple $*$-algebras; therefore, by the finite-dimensional double centralizer theorem, each of them coincides with the commutant of its commutant. We will use this fact repeatedly below; see, for instance, the finite-dimensional form of the bicommutant theorem in standard operator-algebra references such as \cite{conway2019course}.

\subsection{Matrices commuting with the adjacency matrix} \label{secDSCommuting}

Let $A$ be the adjacency matrix of a graph. Define
\[S(A) = \{M : M A = A M, M \text{ is doubly stochastic}\},\]
and let
\[P(A) = \convhull \{P : P A = A P, P \text{ is a permutation matrix}\}.
\]
It follows that $P(A) \subseteq S(A)$, and a graph is \emph{compact} if $P(A) = S(A)$. Interestingly, the fact that $\polAlg$ is already the full matrix algebra for almost all graphs \cite{godsil2012controllable,o2016conjecture} implies that almost all graphs are compact\footnote{This fact has not been spelled out in the literature, as far as we were able to determine.}. 

Let $\ds$ be the operator that takes a set of matrices and returns the subset of all doubly stochastic matrices in it. We may apply $\ds$ at \eqref{commutants}, and if $A = A(X)$, it follows that this chain is sandwiched between $P(A)$ and $S(A)$:
\begin{equation}
P(A) \subseteq \ds (\orbAlg') \subseteq \ds (\polAlg') = S(A). \label{equationClassicalInclusions}
\end{equation}
Evdokimov et al.~\cite{EVDOKIMOV1999247} defined a graph to be \textit{weakly compact} if $P(A) = \ds(\smallAlg')$.

Let $u_{ij}$ be the generators of the quantum automorphism group of $X$, as per Definition~\ref{defQuantumAutomorphismGroup}. A \textit{state} on a $C^*$-algebra is a positive unital linear functional, and the set of all states on a $C^*$-algebra is called its \textit{state space}.
The main object that motivates this paper is the following. 
\begin{definition}\label{defqfa}
The set of quantum fractional automorphisms is defined by
\begin{equation}\label{qsa}
QP(A) = \{ M : M_{ij} = \rho(u_{ij}) \text{ for some } \rho \text { state on } C(\qGroup (X))\}.
\end{equation}
\end{definition}

Note that any matrix $M$ with $M_{ij} = \rho(u_{ij})$ is automatically doubly stochastic, since $u=(u_{ij})$ is a magic unitary.

\subsection{Our results}

We prove several results about the set $QP(A)$. The main point is that $QP(A)$ is the image of the state space of $C(\qGroup(X))$ under the finite list of coordinate evaluations $u_{ij}$, and its commutant recovers the quantum orbital algebra. 

\begin{theoremnocounter}
The set of quantum fractional automorphisms $QP(A)$  is closed, convex and it holds
\begin{equation}
P(A) \subseteq QP(A) \subseteq \ds(\qOrbAlg') \subseteq S(A).
 \label{equationQuantumInclusions}
\end{equation}
\end{theoremnocounter}

In Section~\ref{secQuantumfrac} we prove the theorem above.

In Section~\ref{secCompactness} we discuss when the inclusions displayed above in \eqref{equationClassicalInclusions} and \eqref{equationQuantumInclusions} become equalities. We define notions of (pseudo) classical or quantum compactness, and characterize some of those which are non-trivial. Although one might initially expect that $P(A) \subseteq \ds (\orbAlg')$ and $QP(A) \subseteq \ds(\qOrbAlg')$ are always equalities, it is known \cite{EVDOKIMOV1999247} that the normalized adjacency matrix of the Petersen graph exposes an inequality of the first. Our main result in this section is the theorem below, which states that there are no non-classical examples of quantum compactness.

\begin{theoremnocounter}
If $QP(A) = S(A)$, then $P(A) = S(A)$.
\end{theoremnocounter}

In Section~\ref{secSymmetries} we relate these objects to classical and quantum symmetries. In particular, we discuss vertex-transitivity, generous transitivity, and distance-regularity from the point of view of the equalities $\orbAlg=\polAlg$ and $\qOrbAlg=\polAlg$. At the end of the paper we collect some open questions.

\section{Quantum fractional automorphisms} \label{secQuantumfrac}

Recall the definition of a quantum automorphism group from Definition~\ref{defQuantumAutomorphismGroup} and of the set of quantum fractional automorphisms from Definition~\ref{defqfa}. The next lemma follows from the GNS construction, see for example \cite[Theorem 5.1.1]{murphy2014c}.

\begin{lemma} \label{GNS}
    Let $X$ be a graph and $u$ be the fundamental representation of $\qGroup(X)$. Then for any $\tau$ unital positive linear functional on $C(\qGroup(X))$, there is a (possibly infinite dimensional) Hilbert space $\mathcal{H}$, a representation $\pi:C(\qGroup(X))\to\mathcal B(\mathcal H)$, and a unit vector $v_\tau$ so that 
    \[M=(M_{ij}) = (\tau(u_{ij})) = (\langle \pi(u_{ij}) v_\tau, v_\tau \rangle ) \in QP(A). \]
    The vector $v_\tau$ is called a cyclic vector.
\end{lemma}

If we extend our view to all possible states from the quantum automorphism group, then we can say even more about it. The proof is similar to \cite[Theorem 3.4]{fritz2012tsirelson}.

\begin{lemma}
\label{closedQuantumCompactSet}
    Let $X$ be a graph, $\qGroup(X)$ be its quantum automorphism group and $u$ its fundamental representation. Then the set
    \[QP(A) = \{ M : M_{ij} = \rho(u_{ij}) \text{ for some } \rho \text { state on } C(\qGroup (X))\}\]
    is closed and convex.
\end{lemma}
\begin{proof}
This proof follows from standard theory of $C^*$-algebras and their state spaces; see \cite[Chapter VIII]{conway2019course}.
    
Let $\mathcal S(C(\qGroup(X)))$ denote the state space of $C(\qGroup(X))$. This state space is weak-$*$ compact and convex. Define
\[
    \Phi:\mathcal S(C(\qGroup(X)))\longrightarrow M_n(\mathbb C),
    \qquad
    \Phi(\rho)=(\rho(u_{ij}))_{i,j}.
\]
For each $i,j$, the coordinate map $\rho\mapsto \rho(u_{ij})$ is weak-$*$ continuous, and hence $\Phi$ is continuous. It is also affine. Therefore $QP(A)=\Phi(\mathcal S(C(\qGroup(X))))$ is compact, hence closed, and convex.
\end{proof}

The following theorem will justify our early adoption of the notation $QP(A)$.

\begin{theorem}
\label{thm:CommutantQuantumOrbital}
    Let $X$ be a graph, $u$ be the fundamental representation of its quantum automorphism group and $QP(A)$ the set of quantum fractional automorphisms. Let \(M\in M_V(\mathbb R)\). Then the following are equivalent:
    \begin{itemize}
        \item $M \in \qOrbAlg$.
        \item $[M, u] = 0$.
        \item $M \in QP(A)'$.
    \end{itemize}
    As a consequence, $QP(A)$ is a closed convex set so that 
    \[P(A) \subseteq QP(A) \subseteq \ds(\qOrbAlg') \subseteq S(A).\]
\end{theorem}
\begin{proof} 
The equivalence between the first two conditions is precisely \cite[Theorem 3.10]{LUPINI2020108592} with the only difference being that we are considering a real algebra, but the constraints are all linear, so the result follows: for a real matrix $M$, commuting with the magic unitary is equivalent to being constant on quantum orbitals, that is, to belonging to $\qOrbAlg$.

Assume now that $[M,u]=0$ and let $N=(\rho(u_{ij}))\in QP(A)$. Applying the state $\rho$ entrywise to the equality $Mu=uM$ gives $MN=NM$. Hence $M\in QP(A)'$.

Conversely, suppose that $M\in QP(A)'$ but $Mu\neq uM$. Then for some vertices $i,j$ the element
\[
    a_{ij}=\sum_{k\in V(X)} m_{ik}u_{kj}-\sum_{k\in V(X)}u_{ik}m_{kj}
\]
of $C(\qGroup(X))$ is nonzero. States separate points of a unital $C^*$-algebra (again, see \cite[Chapter VIII]{conway2019course}), so there exists a state $\rho$ such that $\rho(a_{ij})\neq 0$. But for $N=(\rho(u_{ij}))\in QP(A)$, the scalar $\rho(a_{ij})$ is exactly the $(i,j)$-entry of $MN-NM$, contradicting $M\in QP(A)'$. Therefore $Mu=uM$.

The final statement follows from Lemma~\ref{closedQuantumCompactSet}, the fact that every classical automorphism gives a one-dimensional representation of the quantum automorphism group, the fact that any matrix in $QP(A)$ is doubly stochastic, and the equivalence established above.
\end{proof}

\section{Various compactness notions} \label{secCompactness}

Recall from Section~\ref{secCommutants} the chain of inclusions
\[\orbAlg' \subseteq \qOrbAlg' \subseteq \smallAlg' \subseteq \polAlg',\]
and from Section~\ref{secDSCommuting} the relations
\[P(A) \subseteq \ds(\orbAlg') , \quad QP(A) \subseteq \ds(\qOrbAlg') , \quad S(A) = \ds(\polAlg')\]
when restricting the sets to doubly stochastic matrices. 

\subsection{Classical (pseudo) compactness}

Possibly as a warm-up but also introducing new concepts and results, we start by discussing the classical case. A graph is compact if $P(A) = S(A)$. We call a graph \textit{pseudo compact} if 
\[\ds(\orbAlg') = S(A).\]

\begin{lemma} \label{lemmaOrbitalPolytope}
    The set $\orbAlg'$ is linearly spanned by the matrices in $P(A)$. The set $\ds(\orbAlg')$ consists of all non-negative matrices which are affine combinations of matrices in $P(A)$.
\end{lemma}
\begin{proof}
    Since $\orbAlg$ is generated by the orbitals of the automorphism group of $X$, then $\orbAlg'$ is generated by the permutation matrices corresponding to the automorphisms of $X$; see e.g. \cite[Theorem 13.2.1]{ErdosKoRado}. 
    
    Let $M \in \ds(\orbAlg')$. Therefore, $M$ can be written as a linear combination of these permutation matrices. Moreover, since $M$ is doubly stochastic, then $M$ is non-negative and the coefficients of this linear combination must sum to one.

    Conversely, let $M$ be a non-negative matrix which is an affine combination of matrices in $P(A)$. Since $P(A)$ is the convex hull of permutation matrices that commute with $A$, then $M$ can be written as a linear combination of these permutation matrices with coefficients that sum to one. As $M$ is non-negative, $M$ is doubly stochastic and so $M \in \ds(\orbAlg')$.
\end{proof}

When the graph is regular, the condition $\ds(\orbAlg') = S(A)$ is equivalent to the condition that $\orbAlg = \polAlg$.

\begin{lemma}\label{lemmaPseudCompactAlgebra}
    If $\orbAlg = \polAlg$, then $X$ is pseudo compact. If $X$ is regular and pseudo compact, then $\orbAlg = \polAlg$.
\end{lemma}
\begin{proof}
    The first statement is immediate, since if $\orbAlg = \polAlg$, then it holds $\ds(\orbAlg') = \ds(\polAlg') = S(A)$, so $X$ is pseudo compact.

    For the second one, suppose that $X$ is regular and pseudo compact, so $\ds(\orbAlg') = \ds(\polAlg') = S(A)$. Since $X$ is regular, $(1/n)J \in S(A)$. Let $M \in \polAlg'$. Since $J \in \polAlg$, the matrix $M$ commutes with $J$, and therefore there is a scalar $\lambda$ such that $N=M-\lambda I$ has all row and column sums equal to zero. For sufficiently small $\varepsilon>0$, the matrix
    \[
        \frac{1}{n}J+\varepsilon N
    \]
    is doubly stochastic and commutes with $A$. Thus it belongs to $S(A)=\ds(\orbAlg')$. Since $(1/n)J\in S(A)$, it follows that $N\in\orbAlg'$, and consequently $M\in\orbAlg'$. Therefore $\polAlg'\subseteq\orbAlg'$. The reverse inclusion follows from $\polAlg\subseteq\orbAlg$, and hence $\orbAlg'=\polAlg'$. By the double centralizer theorem, $\orbAlg=\polAlg$.
\end{proof}

The hypothesis of regularity is necessary for the nontrivial direction of the above result. The graph $K_{2,3}$ is pseudo compact, but a direct computation gives
$\dim(\orbAlg)=6$ and $\dim(\polAlg)=5$, so $\orbAlg \neq \polAlg$. It follows from this that the orbital algebra of a regular pseudo compact graph is a symmetric association scheme.

\begin{corollary}\label{petersenpc}
    The Petersen graph is pseudo compact but is not compact.
\end{corollary}
\begin{proof}
    The Petersen graph is connected and regular, and is distance-transitive, so $\orbAlg = \polAlg$. However, it is known that the Petersen graph is not compact \cite{EVDOKIMOV1999247}: its normalized adjacency matrix lies in $S(A)$ but not in $P(A)$.
\end{proof}

\subsection{Quantum (pseudo) compactness}

We now define a graph to be \textit{quantum compact} if $QP(A) = S(A)$, and \textit{pseudo quantum compact} if $\ds(\qOrbAlg') = S(A)$.

We start with analogous results to the classical case.

\begin{lemma}
    The set $\qOrbAlg'$ is linearly spanned by the matrices in $QP(A)$. The set $\ds(\qOrbAlg')$ consists of all non-negative matrices which are affine combinations of matrices in $QP(A)$. 
\end{lemma}
\begin{proof}
Let $\mathcal A$ be the unital algebra generated by $QP(A)$. Note that it is closed under the adjoint operation, because for any state $\rho$ the map defined by $\rho' (u_{ij}) = \rho(u_{ji})$ is a state, so it is a $*$-algebra. By the previous Theorem \ref{thm:CommutantQuantumOrbital}, $QP(A)'=\qOrbAlg$. Hence, by the finite-dimensional double centralizer theorem,
\[
    \mathcal A = QP(A)'' = \qOrbAlg'.
\]
Equivalently, $\qOrbAlg'$ is generated as a unital algebra by the matrices in $QP(A)$.

Moreover, the product of two matrices in $QP(A)$ again lies in $QP(A)$. Indeed, if $M^\rho=(\rho(u_{ij}))$ and $M^\sigma=(\sigma(u_{ij}))$, then the convolution state $(\rho\otimes\sigma)\circ \Delta$ satisfies
\[
    ((\rho\otimes\sigma)\circ\Delta)(u_{ij})
    =\sum_k \rho(u_{ik})\sigma(u_{kj}),
\]
which is the $(i,j)$-entry of $M^\rho M^\sigma$. Thus the algebra generated by $QP(A)$ is its linear span.

The second claim follows analogously to the classical case. Any $M \in \ds(\qOrbAlg')$ is non-negative and a linear combination of matrices in $QP(A)$, which are doubly stochastic. So the linear combination must be affine because the rows and columns of $M$ sum to $1$.
\end{proof}

The proof of the following lemma is the same as for the classical case, but using the quantum orbital algebra instead of the orbital algebra.

\begin{lemma}
\label{lemmaPseudQuantCompOrbitAlgeb}
If $\qOrbAlg = \polAlg$, then $X$ is pseudo quantum compact. If $X$ is regular, and if it is pseudo quantum compact, then $\qOrbAlg = \polAlg$.
\end{lemma}
\begin{proof}
    Proof is identical to the proof of Lemma~\ref{lemmaPseudCompactAlgebra}, but using $\qOrbAlg$ instead of $\orbAlg$ and $QP(A)$ instead of $P(A)$.
\end{proof}

Analagous to the classical case, the quantum orbital algebra of a regular pseudo quantum compact graph is a symmetric association scheme. We look again at the Petersen graph and obtain the following separation. 

\begin{corollary}
    The Petersen graph is pseudo quantum compact but it is not quantum compact. 
\end{corollary}

\begin{proof}
We use the fact that the Petersen graph $P$ has no quantum symmetry \cite{schmidt2018petersen}, i.e. the algebra $C(\qGroup(P))$ is commutative. This implies $QP(A)=P(A)$ and $\ds(\qOrbAlg')=\ds(\orbAlg')$ and therefore the statement follows from Corollary \ref{petersenpc}.
\end{proof}

We have seen that the Petersen graph is both pseudo compact and pseudo quantum compact. We will now discuss an example that is pseudo quantum compact, but not pseudo compact. 

\subsection{A non-trivial example of pseudo quantum compactness} \label{secNonTrivialPseudoQuantumCompact}

All graphs which are (pseudo) compact are (pseudo) quantum compact, since $P(A) \subseteq QP(A)$ and $\ds(\orbAlg') \subseteq \ds(\qOrbAlg')$. In this section we show an example of a graph that is pseudo quantum compact but not pseudo compact.

Let $Y$ be the folded halved $8$-cube graph. Its vertex set consists of the pairs $\{x,\mathbf 1+x\}$, where $x\in \mathbb Z_2^8$ has even Hamming weight, and two vertices $\{x,\mathbf 1+x\}$ and $\{y,\mathbf 1+y\}$ are adjacent if either $x+y$ or $x+y+\mathbf 1$ has Hamming weight two.

\begin{definition}
\label{ex:NonTrivialPseudoCompact}
    Let $\Gamma^1$ be the graph whose vertices are the $8$-cliques in the relevant orbit of $\overline{Y}$ under the action of $\mathbb{Z}_2^6 \times A_8$, with two such vertices adjacent when the corresponding cliques intersect in exactly two points.
\end{definition}

This graph is presented in \cite{schmidt2024quantum} and references therein. It is the first known distance-regular graph to have quantum orbital algebra different than the orbital algebra. Its quantum orbital algebra is equal to the smallest coherent algebra. Thus, we get one non-trivial result.

\begin{proposition}
There is a pseudo quantum compact graph which is not pseudo compact.
\end{proposition}
\begin{proof}
The graph $\Gamma^1$ above is distance-regular, and hence its coherent algebra generated by the adjacency matrix coincides with its adjacency algebra: $\smallAlg=\polAlg$. By \cite{schmidt2024quantum}, its quantum orbital algebra satisfies $\qOrbAlg=\smallAlg$, while $\qOrbAlg\neq \orbAlg$. Therefore $\qOrbAlg=\polAlg$, so Lemma~\ref{lemmaPseudQuantCompOrbitAlgeb} implies that $\Gamma^1$ is pseudo quantum compact. On the other hand, since $\orbAlg\neq \polAlg$, Lemma~\ref{lemmaPseudCompactAlgebra} implies that it is not pseudo compact.
\end{proof}

\subsection{Quantum isomorphism preserving pseudo quantum compactness}

Recall the definition of quantum isomorphic graphs \cite{LUPINI2020108592}.

\begin{definition}
Let $X_1$ and $X_2$ be graphs. We say that $X_1$ and $X_2$ are \emph{quantum isomorphic} if there exists a unital $C^*$-algebra $\mathcal{A}$ and a magic unitary $u\in M_n(\mathcal{A})$ such that $A_{X_1}u=uA_{X_2}$, which means $\sum_k(A_{X_1})_{ik}u_{kj}=\sum_k u_{ik}(A_{X_2})_{kj}$ for all $i\in V(X_1), j \in V(X_2)$.
\end{definition}

Quantum isomorphism preserves the quantum orbital algebra in a way that respects the adjacency matrix. Thus, we can show that it preserves pseudo quantum compactness under the hypothesis of regularity.

\begin{theorem}
    Let $X_1, X_2$ be quantum isomorphic regular graphs. Then $X_1$ is pseudo quantum compact if and only if $X_2$ is pseudo quantum compact.
\end{theorem}
\begin{proof}
    Suppose that $X_1$ is pseudo quantum compact, then by Lemma~\ref{lemmaPseudQuantCompOrbitAlgeb} $\mathcal{P}_J(X_1) = \mathcal{O}(X_1)$.  It then follows from \cite[Theorem 4.6]{LUPINI2020108592} that there is an isomorphism between the quantum orbital algebras of $X_1$ and $X_2$ that is linear, multiplicative, sends $A(X_1)$ to $A(X_2)$, $I$ to $I$ and $J$ to $J$.
    
    This means that the isomorphism not only sends the quantum orbital algebra of $X_1$ to the quantum orbital algebra of $X_2$, but the algebra $\mathcal{P}_J(X_1)$ to $\mathcal{P}_J(X_2)$ and so the equality $\mathcal{P}_J(X_2) = \mathcal{O}(X_2)$ also holds. Thus, using again Lemma~\ref{lemmaPseudQuantCompOrbitAlgeb}, the result follows.
\end{proof}

It remains natural to ask whether regularity can be dropped, and whether quantum compactness itself is invariant under quantum isomorphism. 

\subsection{No non-classical quantum compact graphs}

We have seen in Section~\ref{secNonTrivialPseudoQuantumCompact} that there exist pseudo quantum compact graphs that are not pseudo compact. Our work was originally motivated by the question of whether non-classical quantum compact graphs exist, which we now settle in the negative.

\begin{theorem} \label{thmnoqcompact}
    If $QP(A) = S(A)$, then $P(A) = S(A)$.
\end{theorem}

To prove this we use the following standard property of the Grötschel--Lovász--Schrijver theta body defined for a graph $X$: no non-integral vertex of \(\texttt{QSTAB}(X)\) belongs to \(\texttt{TH}(X)\); see \cite[Theorem~1.6]{fujie2002groetschel}. The argument is an elementary consequence of the quadratic-inequality representation of \(\texttt{TH}(X)\), and has already been applied to obtain relevant consequences in quantum information theory in a different setting than ours; see \cite{ramanathan2016no}.

Let $u$ be the fundamental magic unitary from the quantum automorphism group of $X$. Let $\widehat X$ be the graph with vertex set $V(X)\times V(X)$, where we write a vertex of $\widehat X$ as $xa$, where $x,a\in V(X)$, and define adjacency in $\widehat{X}$ by
\begin{equation} xa \sim yb \quad \text{if} \quad u_{xa} u_{yb} = 0. \label{Xhat}
\end{equation}
The results that follow will justify why we may see $\widehat{X}$ as the quantum automorphism incompatibility graph.

For $x,a,i,j\in V(X)$, define
\begin{eqnarray}
    R_x=\{xa : a \in V(X)\},
    \qquad
    C_a=\{xa:x\in V(X)\}, \nonumber \\
    D_{ij}
    =
    \{ia: a \sim j\}
    \cup
    \{bj: b \not\sim i\}. \label{cliques}
\end{eqnarray}

\begin{lemma}
    Let $X$ be a graph and $\widehat{X}$ as defined in \eqref{Xhat}. Then the sets $R_x$, $C_a$ and $D_{ij}$ defined in \eqref{cliques} are cliques of $\widehat{X}$.
\end{lemma}
\begin{proof}
    The fact that $R_x$ and $C_a$ are cliques for all $x$ and $a$ follow from the standard fact that the elements in each row and column of a magic unitary $u$ are orthogonal; see \cite[Remark 1.1.9]{schmidt2020quantum} for instance. 

    For $D_{ij}$, it remains to consider what happens with one vertex $ia$ with $a\sim j$ and one vertex $bj$ with $b\not\sim i$. This case follows from \cite[Proposition 2.1.3]{schmidt2020quantum}. If $A = A(X)$, it follows that 
    \[ A_{aj}=1 \qquad\text{and}\qquad A_{ib}=0. 
    \]
    From $uA=Au$ we have, for all $x,b\in V(X)$, 
    \[ \sum_z u_{xz}A_{zb} = \sum_z A_{xz}u_{zb}. 
    \]
    Multiplying this equality on the left by $u_{xa}$ and on the right by $u_{yb}$ gives 
    \[ A_{ab}u_{xa}u_{yb} = A_{xy}u_{xa}u_{yb}, 
    \]
    where we used the row and column orthogonality relations of the magic unitary. Hence, if $A_{xy}\neq A_{ab}$, then $u_{xa}u_{yb}=0$. Therefore every two distinct vertices of $D_{ij}$ are adjacent in $\widehat X$, and $D_{ij}$ is a clique.
\end{proof}

We now relate $QP(A)$ and $S(A)$ with convex corners defined by $\widehat{X}$.

For ease of mind when thinking in terms of matrices and vectors, let us consider $\operatorname{vec}: M_V(\mathbb R) \to \mathbb{R}^{V \times V}$ to be the standard ``column stacking operator''. For the result below, we use the standard presentation of the theta body as a projection of lifted positive semidefinite matrices satisfying certain linear constraints; see \cite[Lemma 2.17]{lovasz1991cones}: if $Y$ is a graph, then $v \in \texttt{TH}(Y)$ if $v_a = \widehat{M}_{0a}$, where $\widehat{M}$ is any positive semidefinite matrix indexed by $\{0\}\cup V(Y)$, with $(0,0)$-entry equal to $1$, with $(0,a)$-entry equal to the $(a,a)$-entry for all $a \in V(Y)$, and with $(a,b)$-entry equal to $0$ whenever $a\sim b$ in $Y$.

\begin{lemma}
\label{lemmaQPThetaBody}
    If $M\in QP(A)$, then $\operatorname{vec}(M)\in \texttt{TH}(\widehat X)$.
\end{lemma}
\begin{proof}
    Let $M_{xa}=\rho(u_{xa})$ for a state $\rho$ on $C(\qGroup(X))$. By the GNS construction (Lemma~\ref{GNS}), there are a representation $\pi$ on a Hilbert space $\mathcal H$ and a unit vector $\psi$ such that
    \[
        M_{xa}=\langle \pi(u_{xa})\psi,\psi\rangle .
    \]
    Consider the Gram matrix $\widehat M$ of the set of vectors $\{\psi\} \cup \{\pi(u_{xa})\psi\}_{xa}$, which we index by $\{0\}\cup V(\widehat X)$. Its entries are
    \[
        \widehat M_{00}=1,\qquad
        \widehat M_{0,xa}=\widehat M_{xa,0}=M_{xa},
    \]
    and
    \[
        \widehat M_{xa,yb}
        =
        \langle \pi(u_{xa})\psi,\pi(u_{yb})\psi\rangle .
    \]
    This matrix is positive semidefinite. Moreover,
    \[
        \widehat M_{xa,xa}
        =
        \langle \pi(u_{xa})\psi,\pi(u_{xa})\psi\rangle
        =
        \langle \pi(u_{xa})\psi,\psi\rangle
        =
        M_{xa},
    \]
    because each $u_{xa}$ is a projection. Finally, if $xa\sim yb$ in $\widehat X$, then $u_{xa}u_{yb}=0$, and therefore $\widehat M_{xa,yb}=0$. Hence $\operatorname{vec}(M)\in \texttt{TH}(\widehat X)$.
\end{proof}

We next record a useful description of $S(A)$ in terms of clique equalities defined by the cliques in \eqref{cliques}.

\begin{lemma}
\label{lemmaSAcliqueEqualities}
    Let $X$ be a graph, $A=A(X)$. Let $M$ be a non-negative matrix and let $f=\operatorname{vec}(M)$. Then $M\in S(A)$ if and only if
    \[
        f(R_x)=1 \quad \text{for all }x\in V(X),
    \]
    \[
        f(C_a)=1 \quad \text{for all }a\in V(X),
    \]
    and
    \[
        f(D_{ij})=1 \quad \text{for all }i,j\in V(X).
    \]
\end{lemma}
\begin{proof}
    The first two families of equalities say precisely that $M$ has all row and column sums equal to $1$, which along with nonnegativity is equivalent to $M$ being doubly stochastic. Assuming these equalities, we have, for all $i,j\in V(X)$,
    \[
        f(D_{ij})
        =
        \sum_{a\sim j}M_{ia}
        +
        \sum_{b\not\sim i}M_{bj}.
    \]
    Since the $j$-th column of $M$ sums to $1$, the equality $f(D_{ij})=1$ is equivalent to
    \[
        \sum_{a\sim j}M_{ia}
        =
        \sum_{b\sim i}M_{bj}.
    \]
    This is exactly the $(i,j)$-entry of the equation $MA=AM$. Hence the displayed equalities are equivalent to $M$ being doubly stochastic and commuting with $A$, that is, to $M\in S(A)$.
\end{proof}

We use the following presentation of the theta body; see \cite[Theorem~1.6]{fujie2002groetschel}. Let $\diag: M_V(\mathbb R) \to \mathbb{R}^{V}$ be the operator defined by $\diag(M)_a = M_{aa}$, for $a \in V$. If $Y$ is a graph, then
\[
\texttt{TH}(Y)=
\left\{
v\in\mathbb R_+^{V(Y)}:
v^\T B v - v^\T \diag(B) \leq 0
\text{ for all } B \in \mathcal M(Y)
\right\},
\]
where
\[
\mathcal M(Y)=
\{B\in\mathbb S_+^{V(Y)}: B_{ab}=0
\text{ whenever }a\neq b\text{ and }a\not\sim b\}.
\]

\begin{lemma}
\label{lemmaFractionalVertexNotTheta}
    Let $M$ be a vertex of $S(A)$. If $M$ is not a permutation matrix, then
    \[
        \operatorname{vec}(M)\notin \texttt{TH}(\widehat X).
    \]
\end{lemma}
\begin{proof}
    Let $f=\operatorname{vec}(M)$. Since $M$ is doubly stochastic, if $M$ is not a permutation matrix, then some entry of $f$ lies strictly between $0$ and $1$.

    Since $M$ is a vertex of $S(A)$, there is a collection of linearly independent active constraints defining $f$. These constraints consist of some non-negativity constraints $f_v=0$ and some of the equalities
    \[
        f(R_x)=1,\qquad f(C_a)=1,\qquad f(D_{ij})=1.
    \]
    Choose such a linearly independent collection whose normal vectors span $\mathbb R^{V(\widehat X)}$ (which is possible since $f$ is a vertex). Let $Z$ be the set of vertices $v$ for which the constraint $f_v=0$ is chosen, and let $\mathcal C$ be the chosen collection of cliques among the $R_x,C_a,D_{ij}$.

    Define
    \[
        B_0=
        \sum_{v\in Z} e_v e_v^\T
        +
        \sum_{C\in\mathcal C}\mathds 1_C\mathds 1_C^\T,
    \]
    where $\mathds 1_C$ is the characteristic vector of the clique $C$. By construction, $B_0$ is positive definite, as it is a sum of scaled projectors onto vectors forming a basis. 
    
    Choose a coordinate $p$ such that $0<f_p<1$. For sufficiently small $\varepsilon>0$, the matrix
    \[
        B=B_0-\varepsilon e_p e_p^\T
    \]
    is still positive semidefinite. Moreover, $B\in\mathcal M(\widehat X)$, because each matrix $\mathds 1_C\mathds 1_C^\T$ is supported on the diagonal and on pairs of adjacent vertices of $\widehat X$.

    We now evaluate the theta-body inequality at $f$. The non-negativity constraints in $Z$ contribute zero, since $f_v=0$ for all $v\in Z$. Each chosen clique equality contributes
    \[
        f^\T\mathds 1_C\mathds 1_C^\T f
        -
        \sum_{v\in V(\widehat X)}(\mathds 1_C\mathds 1_C^\T)_{vv}f_v
        =
        f(C)^2-f(C)
        =
        1^2-1
        =
        0.
    \]
    Therefore
    \[
        f^\T B f-\sum_v B_{vv}f_v
        =
        -\varepsilon f_p^2+\varepsilon f_p
        =
        \varepsilon f_p(1-f_p)
        >
        0.
    \]
    This violates the defining inequalities of $\texttt{TH}(\widehat X)$. Hence $f\notin \texttt{TH}(\widehat X)$.
\end{proof}

The proof of Theorem~\ref{thmnoqcompact} now follows immediately.

\begin{proof}[Proof of Theorem~\ref{thmnoqcompact}]
    Suppose that $QP(A)=S(A)$, and let $M$ be a vertex of $S(A)$. If $M$ is not a permutation matrix, then Lemma~\ref{lemmaFractionalVertexNotTheta} gives
    \[
        \operatorname{vec}(M)\notin \texttt{TH}(\widehat X).
    \]
    On the other hand, since $M\in S(A)=QP(A)$, Lemma~\ref{lemmaQPThetaBody} gives
    \[
        \operatorname{vec}(M)\in \texttt{TH}(\widehat X),
    \]
    a contradiction. Thus every vertex of $S(A)$ is a permutation matrix.

    Since every permutation matrix in $S(A)$ is an automorphism of $X$, every vertex of $S(A)$ belongs to $P(A)$. As $P(A)\subseteq S(A)$ and $S(A)$ is the convex hull of its vertices, we conclude that
    \[
        P(A)=S(A).
    \]
\end{proof}

\section{Symmetries and regularities} \label{secSymmetries}

In this section we relate the properties of compactness and pseudo compactness, classical or quantum, to symmetries of the graph. We start with vertex-transitivity, and then we will talk about distance-regularity and equitable partitions.

\subsection{Vertex transitivity}

A graph is vertex-transitive if its automorphism group acts transitively on its vertices. A known result due to Godsil \cite{GODSIL1997Compact} and Tinhofer \cite{TINHOFER1991253} is that if a graph is regular and compact, then it is vertex-transitive. We can extend this result to the quantum case.

A graph is quantum vertex-transitive if for all $a,b \in V(X)$, the entry $u_{ab}$ from the fundamental representation of the quantum automorphism group of $X$ is not zero; see \cite{LUPINI2020108592}.

\begin{proposition}
    If $X$ is regular and pseudo quantum compact, then $X$ is quantum vertex-transitive.
\end{proposition}
\begin{proof}
    Since the graph is regular, $(1/n)J \in S(A)$, and since the graph is pseudo quantum compact, $(1/n)J$ is an affine combination of matrices in $QP(A)$. In particular, for all $a,b \in V(X)$, there is a matrix $M \in QP(A)$ such that $M_{ab} > 0$. Since $M\in QP(A)$, there is a state $\rho$ such that $M_{ab} = \rho(u_{ab}) > 0$, so $u_{ab} \neq 0$.
\end{proof}

Similarly, one obtains that a regular, pseudo compact graph $X$ is vertex-transitive.

\subsection{Generous transitivity}

In \cite{GODSIL1997Compact}, it is shown that regular compact graphs are not only vertex-transitive, but also have a \emph{generously transitive} automorphism group. A permutation group on a set $\Omega$ is generously transitive if, for every $a,b\in\Omega$, there is a group element $g$ such that $g(a)=b$ and $g(b)=a$. Equivalently, all orbitals of the group are symmetric.

Godsil's original proof uses representation theory of the automorphism group. We can recover an analogous conclusion using the coherent-algebra framework.

\begin{theorem}
    Let $X$ be a regular graph with $r$ distinct eigenvalues. If $X$ is pseudo compact,  then $\aut(X)$ is generously transitive with rank $r$, i.e., with $r$ orbitals.
\end{theorem}

\begin{proof}
    By Lemma \ref{lemmaPseudCompactAlgebra}, if the graph is regular, then pseudo compactness is equivalent to $\orbAlg=\polAlg$.
    
    This immediately implies that $\orbAlg$ is symmetric. It follows that the orbital algebra being symmetric is equivalent to its group being generously transitive; see \cite[section 2.9]{BCN} for reference.

    The rank of the group being equal to the number of eigenvalues follows from the dimension of the orbital algebra and algebra generated by $A$ being the same.
\end{proof}

We also generalize this to the quantum setting. The definition below can also be found in \cite[Section 6.9]{mariiaSobchukPhd}.

\begin{definition}
   Let $X$ be a graph. We say that its quantum automorphism group $\qGroup (X)$ is generously transitive if
   \[u_{ab}u_{ba}\neq 0, \quad \text{for all }a,b\in V(X).\]
\end{definition}

\begin{corollary}
    Let $X$ be a regular graph with $r$ distinct eigenvalues. If $X$ is pseudo quantum compact,  then its quantum automorphism group $\qGroup(X)$ is generously transitive with $r$ orbitals.
\end{corollary}
\begin{proof}
    By Lemma~\ref{lemmaPseudQuantCompOrbitAlgeb}, pseudo quantum compactness is equivalent to $\qOrbAlg=\polAlg$. As above, connected regularity implies that $\polAlg$ is generated by the symmetric matrix $A$ and has dimension $r$. Hence all quantum orbital matrices are symmetric, which is equivalent to $u_{ab}u_{ba}\neq 0$ for all $a,b$. The number of quantum orbitals is $\dim \qOrbAlg=r$.
\end{proof}

\subsection{Distance-regularity and equitable partitions}

Let $X$ be a graph. A partition $C_1, \cdots, C_s $ of $V(X)$ is said to be \emph{equitable} if for every ordered pairs $(C_i, C_j)$ the number of vertices in $C_j$ adjacent to a fixed vertex in $C_i$ only depends on $i$ and $j$. A partition is called a \emph{distance partition} with respect to a vertex $v$ if the cells are ordered such that vertices in cell $C_i$ are at distance $i$ from $v$.

A graph $X$ is called \textit{distance-regular} if the distance partition with respect to any of its vertices is equitable. Moreover, it is \emph{distance-transitive} if it is distance-regular and for any pairs $(a, b), (a', b') \in V(X)^2$ with distance between $a$ and $b$ equal to the distance between $a'$ and $b'$, then there is an automorphism that sends $a$ to $a'$ and $b$ to $b'$; see \cite{BCN} for a standard reference.

It is known that not all distance-regular graphs are distance-transitive. In \cite{GODSIL1997Compact}, it was shown that in compact graphs, equitable partitions are related to partitions defined by the orbits of the graph. Moreover, in Corollary 1.4 therein the author shows that if the graph is distance-regular and compact it is distance-transitive. We strengthen this result using our framework.

\begin{proposition}
    Let $X$ be a distance-regular graph. Then it is distance-transitive if and only if it is pseudo compact.
\end{proposition}
\begin{proof}
    Let $d$ be the diameter of $X$. For a distance-regular graph, the distance matrices form the Bose--Mesner algebra, and this algebra is generated by $A$; hence $\polAlg=\smallAlg$ and has dimension $d+1$. If $X$ is distance-transitive, then the orbitals of $\aut(X)$ are precisely the distance relations, so $\orbAlg$ also has dimension $d+1$ and equals $\polAlg$.

    Conversely, if $X$ is pseudo compact, Lemma~\ref{lemmaPseudCompactAlgebra} gives $\orbAlg=\polAlg$. Since $\polAlg$ is spanned by the distance matrices, the orbitals must coincide with the distance relations. Thus $X$ is distance-transitive.
\end{proof}

Thus, on distance-regular graphs, pseudo compactness characterizes distance-transitivity, whereas compactness is strictly stronger: the Petersen graph is distance-transitive and pseudo compact but not compact.

We therefore find it natural to define the following:
\begin{definition}
    A graph $X$ is said to be \emph{quantum distance-transitive} if it is distance-regular and 
    $\polAlg = \qOrbAlg$.
\end{definition}

An equivalent formulation, closer to the classical definition, asks that whenever $\operatorname{dist}(a,b)=\operatorname{dist}(a',b')$, the pairs $(a,b)$ and $(a',b')$ lie in the same quantum orbital, that is, $u_{aa'}u_{bb'}\neq 0$ in the convention of \cite{LUPINI2020108592}.

\begin{proposition}
    A distance-regular graph $X$ is quantum distance-transitive if and only if it is pseudo quantum compact.
\end{proposition}
\begin{proof}
    Since $X$ is distance-regular, $\polAlg=\smallAlg$. The assertion follows from Lemma~\ref{lemmaPseudQuantCompOrbitAlgeb} and the definition of quantum distance-transitivity.
\end{proof}

All distance-transitive graphs are quantum distance-transitive, but the converse is not necessarily true. The graph $\Gamma^1$ in Definition \ref{ex:NonTrivialPseudoCompact} is one non-trivial example.

\section{Conclusion and open questions}

\label{sec:conclusion}

This work connects the notion of compactness of a graph with quantum symmetries, and the structure of the coherent algebras containing the adjacency matrix. A main tool is the repeated use of the double centralizer theorem, simplifying the proof of known results and allowing for their extension to the quantum case. We defined a new object, the set of doubly stochastic matrices that are obtained by evaluating the fundamental magic
unitary of the quantum automorphism group on states, and showed how it is related to the quantum orbital algebra and the quantum automorphism group of the graph. We gave an example of a graph that is pseudo quantum compact but not pseudo compact and we introduced connections to quantum symmetries and regularities of the graph.

Our main result is that a quantum compact graph is (classically) compact.

We leave the following open questions:

\begin{itemize}
    \item Is $QP(A)$ always polytope? If it is not, can one characterize under which conditions, be it through the structure of the graph or its automorphisms, is this set a polytope?
    \item We know that deciding isomorphism for compact graphs can be done in polynomial time. We find it interesting to ask if quantum isomorphism can be decided (efficiently or not) for compact graphs.
    \item We would like to understand better what quantum isomorphisms gives us. We know that it preserves pseudo-compactness and the quantum orbital algebra structure. It also preserves doubly stochastic matrices inside the quantum orbital algebra. But does it ''preserve" its commutant, meaning it sends quantum fractional automorphisms of one graph to quantum fractional automorphisms of the other? Will it be a polytope for one graph if and only if it is for the other?
\end{itemize}

\section*{Acknowledgements}

We acknowledge Thiago Assis, who participated in the early stages of this project and contributed to some of the initial ideas. Pedro Baptista and Gabriel Coutinho were supported by CNPq and FAPEMIG. Chris Godsil acknowledges support from NSERC grant RGPIN-9439. Simon Schmidt acknowledges support from the Deutsche Forschungsgemeinschaft
(DFG, German Research Foundation) under Germany’s Excellence Strategy – EXC 2092 CASA
– 390781972.

\bibliographystyle{unsrt}
\bibliography{qcomgraph}

\end{document}